\newif\ifanon
\newif\ifsubmission

\submissionfalse

\ifsubmission
  \documentclass[a4paper,11pt]{article}
  \usepackage{a4wide}
  \date{}
  \anontrue
  \author{No Author Given}
\else
  \documentclass[runningheads]{llncs}
  \anonfalse
\fi

\usepackage[utf8]{inputenc}
\usepackage[english]{babel}
\usepackage[T1]{fontenc}
\usepackage{amssymb,amsmath}
\usepackage{graphicx}
\usepackage{hyperref}
\usepackage{url}
\ifsubmission
\usepackage{amsthm}
\fi
\usepackage{hyperref}

\newcommand{\bF}{{\mathbb F}}
\newcommand{\F}{{\mathbb F}}

\newcommand{\cS}{\ensuremath{\mathcal S}}
\newcommand{\cR}{\ensuremath{\mathcal R}}
\DeclareMathOperator{\PGL}{PGL}

\def\bP{\ensuremath{\mathbb{P}}}
\def\bZ{\ensuremath{\mathbb{Z}}}
\def\cP{\ensuremath{\mathcal{P}}}
\def\cL{\ensuremath{\mathcal{L}}}
\def\cH{\ensuremath{\mathcal{H}}}

\newtheorem{theo}{Theorem}

\newtheorem{prop}[theo]{Proposition}
\newtheorem{cor}[theo]{Corollary}

\newtheorem{definition}[theo]{Definition}

\newtheorem{heuristic}[theo]{Heuristic}
\def\coloneq{\mathrel{:=}}
\makeatletter
\renewenvironment{proof}[1][\proofname]{\par
  \normalfont \topsep 6pt plus 6pt\relax
  \trivlist
  \item[\hskip\labelsep
        \itshape
    #1\@addpunct{.}]\ignorespaces
}{%
  \qed\endtrivlist\@endpefalse
}
\makeatother

\ifanon
\else
\titlerunning{A quasi-polynomial algorithm for DLP in small
characteristic finite fields}
\fi
\title{A quasi-polynomial algorithm for discrete logarithm in finite fields of small characteristic}
\ifanon\else
\hypersetup{
    pdfauthor={Razvan Barbulescu, Pierrick Gaudry, Antoine Joux, Emmanuel
    Thomé},
    pdftitle={A quasi-polynomial algorithm for discrete logarithm in finite
    fields of small characteristic},
    pdfborder={0 0 0}
}

\author{Razvan~Barbulescu\inst{1} \and
        Pierrick~Gaudry\inst{1} \and
        Antoine~Joux\inst{2,3} \and
        Emmanuel~Thomé\inst{1}}
\date{}
\institute{Inria, CNRS, University of Lorraine, France
    \and
    Cryptology Chair, Foundation UPMC -- LIP~6, CNRS UMR 7606, Paris, France
    \and
    CryptoExperts, Paris, France}
\fi

\begin{document}

\maketitle
\begin{abstract}
The difficulty of computing discrete logarithms in fields~$\F_{q^k}$
depends on the relative sizes of $k$ and $q$. Until recently all the
cases had a sub-exponential complexity of type $L(1/3)$, similar to the
factorization problem. In 2013, Joux designed a new algorithm with a
complexity of $L(1/4+\epsilon)$ in small characteristic. In the same
spirit, we propose in this article another heuristic algorithm that provides a
quasi-polynomial complexity when $q$ is of size at most comparable with
$k$. By quasi-polynomial, we mean a runtime of $n^{O(\log n)}$ where $n$
is the bit-size of the input.  For larger values of $q$ that stay below
the limit $L_{q^k}(1/3)$, our algorithm loses its quasi-polynomial
nature, but still surpasses the Function Field Sieve. 
\end{abstract}

\section{Introduction}

The discrete logarithm problem (DLP) was first proposed as a hard
problem in cryptography in the seminal article of Diffie and
Hellman~\cite{DiHe76}. Since then, together with factorization, it has
become one of the two major pillars of public key cryptography. As a
consequence, the problem of computing discrete logarithms has
attracted a lot of attention. From an exponential algorithm in $1976$,
the fastest DLP algorithms have been greatly improved during the past
$35$ years. A first major progress was the realization that the DLP in
finite fields can be solved in subexponential time, i.e. $L(1/2)$
where $L_N(\alpha)=\exp\left(O((\log N)^\alpha(\log\log
  N)^{1-\alpha})\right)$. The next step further reduced this to a
heuristic $L(1/3)$ running time in the full range of finite fields,
from fixed characteristic finite fields to prime
fields~\cite{Adl79,Cop84,Gor93,Adl94,JoLe06,JLVS07}.

Recently, practical and theoretical advances have been
made~\cite{Jo13faster,GGMZ13,Joux13} with an
emphasis on small to medium characteristic finite fields and composite
degree extensions. The most general and efficient
algorithm~\cite{Joux13} gives a complexity of $L(1/4+o(1))$ when the
characteristic is smaller than the square root of the extension
degree. Among the ingredients of this approach, we find the use of a
very particular representation of the finite field; the use of the
so-called {\em systematic equation}\footnote{While the terminology is
similar, no parallel is to be made with the systematic equations as
defined in early works related to the computation discrete logarithms in
$\bF_{2^n}$, as~\cite{BlFuMuVa84}.}; and the use of algebraic
resolution of bilinear polynomial systems in the individual logarithm
phase.

In this work, we present a new discrete logarithm algorithm, in
the same vein as in~\cite{Joux13} that uses an asymptotically more
efficient descent approach. The main result gives a {\it
  quasi-polynomial} heuristic complexity for the DLP in finite fields
of small characteristic. By quasi-polynomial, we mean a complexity of
type $n^{O(\log n)}$ where $n$ is the bit-size of the cardinality of
the finite field. Such a complexity is smaller than any
$L(\epsilon)$ for $\epsilon>0$. It remains super-polynomial
in the size of the input, but offers a major asymptotic improvement
compared to $L(1/4+o(1))$.

The key features of our algorithm are the following.
\begin{itemize}
\item We keep the field representation and the systematic equations of~\cite{Joux13}.
\item The algorithmic building blocks are elementary. In particular,
  we avoid the use of Gröbner basis algorithms.
\item The complexity result relies on three key heuristics:
the existence of a polynomial representation of the appropriate
form; the fact that the smoothness probabilities of some non-uniformly
distributed
polynomials are similar to the probabilities for uniformly random
polynomials of the same degree; and the linear independence of some
finite field elements related to the action of $\PGL_2(\bF_{q})$.
\end{itemize}

The heuristics are very close to the ones used in~\cite{Joux13}. In
addition to the arguments in favor of these heuristics already given
in~\cite{Joux13}, we performed some experiments to validate them on
practical instances.  \medskip

Although we insist on the case of finite fields of small
characteristic, where quasi-polynomial complexity is obtained, our new
algorithm improves the complexity of discrete logarithm computations in a
much larger range of finite fields.

More precisely, in finite fields of the form $\bF_{q^k}$, where $q$
grows as $L_{q^k}(\alpha)$, the complexity becomes
$L_{q^k}(\alpha+o(1))$. As a consequence, our algorithm is
asymptotically faster than the Function Field Sieve algorithm in
almost all the range previously covered by this algorithm. Whenever 
$\alpha<1/3$, our new algorithm offers the smallest complexity. For
the limiting case $L(1/3,c)$, the Function Field Sieve remains more
efficient for small values of $c$, and the Number Field Sieve is better
for large values of $c$ (see~\cite{JLVS07}).
\bigskip

This article is organized as follows. In Section~\ref{sec:main}, we state
the main result, and discuss how it can be used to design a complete
discrete logarithm algorithm. In Section~\ref{sec:csq}, we analyze how
this result can be interpreted for various types of finite fields,
including the important case of fields of small characteristic.
Section~\ref{sec:descent-one-step} is devoted to the description of our
new algorithm. It relies on heuristics that are discussed in
Section~\ref{sec:heur}, from a theoretical and a practical point of view.
Before getting to the conclusion, in Section~\ref{sec:improvement}, we
propose a few variants of the algorithm.

\section{Main result}
\label{sec:main}

We start by describing the setting in which our algorithm applies. It is
basically the same as in~\cite{Joux13}: we need a large enough subfield,
and we assume that a sparse representation can be found.  This is
formalized in the following definition.

\begin{definition}
    A finite field $K$ admits a {\em sparse medium subfield representation} if
    \begin{itemize}
        \item it has a subfield of $q^2$ elements for a prime power $q$,
		i.e. $K$ is isomorphic to $\F_{q^{2k}}$ with $k\geq1$;
        \item there exist two polynomials $h_0$ and $h_1$ over
            $\F_{q^2}$ of small degree, such that $h_1X^q-h_0$ has a
            degree $k$ irreducible factor.
    \end{itemize}
\end{definition}

In what follows, we will assume that all the fields under consideration
admit a sparse medium subfield representation. Furthermore, we assume that
the degrees of the polynomials $h_0$ and $h_1$ are uniformly bounded by a
constant $\delta$.  Later, we will provide heuristic arguments for the
fact that any finite field of the form $\F_{q^{2k}}$ with $k \le q+2$
admits a sparse medium subfield representation with polynomials $h_0$ and
$h_1$ of degree at most 2.  But in fact, for our result to hold, allowing
the degrees of $h_0$ and $h_1$ to be bounded by any constant $\delta$
independent of $q$ and $k$ or even allowing $\delta$ to grow
slower than $O(\log q)$ would be sufficient.

In a field in sparse medium subfield representation, elements will
always be represented as polynomials of degree less than $k$ with
coefficients in $\F_{q^2}$. When we talk about the discrete logarithm of
such an element, we implicitly assume that a basis for this discrete
logarithm has been chosen, and that we work in a subgroup whose order has
no small irreducible factor (we refer to the Pohlig-Hellman
algorithm~\cite{PoHe78} to limit ourselves to this case).

\begin{prop}\label{prop:onestep}
    Let $K=\F_{q^{2k}}$ be a finite field that admits a sparse medium subfield
    representation.
    Under the heuristics explained below, there exists an algorithm whose
    complexity is polynomial in $q$ and $k$ and which can be used for the
    following two tasks. 

    \begin{enumerate}
    \item 
        Given an element of $K$ represented by a polynomial
        $P\in\F_{q^2}[X]$ with $2\leq \deg P\leq k-1$,
        the algorithm returns an expression of
        $\log P(X)$ as a linear combination of at most $O(kq^2)$
        logarithms $\log P_i(X)$ with $\deg P_i \leq \lceil
        \frac12 \deg P\rceil$ and of $\log h_1(X)$.

    \item
        The algorithm returns the logarithm of $h_1(X)$ and 
        the logarithms of all the elements of $K$
        of the form $X+a$, for $a$ in $\F_{q^2}$.  
     \end{enumerate}
\end{prop}

Before the presentation of the algorithm, which is made in Section~\ref{sec:descent-one-step}, we explain how to use it as a building block for a complete discrete logarithm algorithm.

Let $P(X)$ be an element of $K$ for which we want to compute the discrete
logarithm. Here $P$ is a polynomial of degree at most $k-1$ and with
coefficients in $\F_{q^2}$. We start by applying the algorithm of Proposition~\ref{prop:onestep} to $P$. We obtain a relation of the form
$$ \log P = e_0 \log h_1 + \sum e_i \log P_i,$$
where the sum has at most $\kappa q^2 k$ terms for a constant
$\kappa$ and the $P_i$'s have degree at most $\lceil \frac12 \deg P\rceil$.
Then, we apply
recursively the algorithm to the $P_i$'s, thus creating a descent
procedure where at each step, a given element $P$  is expressed as a
product of elements, whose degree is at most half the degree of $P$
(rounded up)  and the arity of the descent tree is in $O(q^2 k)$.

At the end of the process, the logarithm of $P$ is expressed as a linear
combination of the logarithms of $h_1$ and of the linear polynomials,
for which the logarithms are computed with the algorithm in
Proposition~\ref{prop:onestep} in its second form.

We are left with the complexity analysis of the descent
process. Each internal node of the descent tree corresponds to one application of
the algorithm of Proposition~\ref{prop:onestep}, therefore each internal
node has a cost which is bounded by a polynomial in $q$ and~$k$. The total cost
of the descent is therefore bounded by the number of nodes in the descent
tree times a polynomial in $q$ and $k$.  The depth of the descent tree is in
$O(\log k)$. The number of nodes of the tree is then less than or equal to
its arity raised to the power of its depth, which is $(q^2
k)^{O(\log k)}$.  Since any polynomial in $q$ and
$k$ is absorbed in the $O()$ notation in the exponent, we obtain the
following result.

\begin{theo}\label{thm}
    Let $K=\F_{q^{2k}}$ be a finite field that admits a sparse medium
    subfield representation. Assuming the same heuristics as in
    Proposition~\ref{prop:onestep}, any discrete logarithm in $K$ can be
    computed in a time bounded by 
    $$ \max(q,k)^{O(\log k)}.$$
\end{theo}

\section{Consequences for various ranges of parameters}
\label{sec:csq}

We now discuss the implications of Theorem~\ref{thm} depending on the
properties of the finite field $\F_Q$ where we want to compute discrete
logarithms in the first place. The complexities will be expressed in
terms of $\log Q$, which is the size of the input.

Three cases are considered. In the first one, the finite field admits a
sparse medium subfield representation, where $q$ and $k$ are almost
equal. This is the optimal case. Then we consider the case where the
finite field has small (maybe constant) characteristic. And finally, we
consider the case where the characteristic is getting larger so that the
only available subfield is a bit too large for the algorithm to have an
optimal complexity.

In the following, we always assume that for any field of the form
$\F_{q^{2k}}$, we can find a sparse medium subfield representation.

\subsection{Case where the field is $\F_{q^{2k}}$, with $q\approx k$}

The finite fields $\F_Q = \F_{q^{2k}}$ for which $q$ and $k$ are almost
equal are tailored for our algorithm. In that case, the complexity of
Theorem~\ref{thm} becomes $q^{O(\log q)}$. Since $Q \approx q^{2q}$, we
have $q=(\log Q)^{O(1)}$. This gives an expression of the form
$2^{O\left((\log \log Q)^2\right)}$, which is sometimes called
quasi-polynomial in complexity theory.

\begin{cor}\label{cor1}
    For finite fields of cardinality $Q = q^{2k}$ with $q+O(1)\geq k$
    and $q=(\log Q)^{O(1)}$,
    there exists a heuristic algorithm for computing discrete logarithms
    in quasi-polynomial time
    $$ 2^{O\left((\log \log Q)^2\right)}.$$
\end{cor}

We mention a few cases which are almost directly covered by
Corollary~\ref{cor1}. First, we consider the case where $Q=p^n$ with
$p$ a prime bounded by $(\log
Q)^{O(1)}$, and yet large enough so that $n \le (p+\delta)$. In this
case $\F_Q$, or possibly $\F_{Q^2}$ if $n$ is odd, can be represented in
such a way that Corollary~\ref{cor1} applies.

Much the same can be said in the case where $n$ is composite and factors
nicely, so that $\F_Q$ admits a large enough subfield $\F_q$ with
$q=p^m$. This can be used to solve certain discrete logarithms in, say,
$\bF_{2^n}$ for adequately chosen $n$ (much similar to records tackled
by~\cite{record1778,record1971,record4080,record6120,record6168}).

\subsection{Case where the characteristic is polynomial in the input size}

Let now $\F_Q$ be a finite field whose characteristic $p$ is bounded by
$(\log Q)^{O(1)}$, and let $n=\log Q / \log p$, so that $Q = p^n$. While
we have seen that Corollary~\ref{cor1} can be used to treat some cases,
its applicability might be hindered by the absence of an appropriately
sized subfield: $p$ might be as small as $2$,
and $n$ might not factor adequately. In those cases, we use the same strategy as
in~\cite{Joux13} and embed the discrete logarithm problem in $\F_Q$ into
a discrete logarithm problem in a larger field.

Let $k$ be $n$ if $n$ is odd and $n/2$ if $n$ is even. Then, we set
$q = p^{\lceil \log_p k \rceil}$, and we work in the field $\F_{q^{2k}}$.
By construction this field contains $\F_Q$ (because $p|q$ and $n|2k$) and
it is in the range of applicability of Theorem~\ref{thm}. Therefore,
one can solve a discrete logarithm problem in $\F_Q$ in time
$\max(q, k)^{O(\log k)}$. Rewriting this complexity in terms of $Q$, we get
$\log_p(Q)^{O(\log\log Q)}$. And finally, we get a similar complexity
result as in the previous case. Of course, since we had to embed in a
larger field, the constant hidden in the $O()$ is larger than for
Corollary~\ref{cor1}.

\begin{cor}\label{cor2}
    For finite fields of cardinality $Q$ and characteristic bounded by
    $\log(Q)^{O(1)}$, there exists a heuristic algorithm for 
    computing discrete logarithms in quasi-polynomial time
    $$ 2^{O\left((\log \log Q)^2\right)}.$$
\end{cor}

We emphasize that the case $\F_{2^n}$ for a prime $n$ corresponds to
this case. A direct consequence of Corollary~\ref{cor2} is that discrete logarithms in $\F_{2^n}$ can be computed in
quasi-polynomial time $2^{O((\log n)^2)}$.

\subsection{Case where $q = L_{q^{2k}}(\alpha)$}
If the characteristic of the base field is not so small compared to
the extension degree, the complexity of our algorithm does not keep
its nice quasi-polynomial form. However, in almost the whole range of
applicability of the Function Field Sieve algorithm, our algorithm is
asymptotically better than FFS.

We consider here finite fields that can be put into the form $\F_Q =
\F_{q^{2k}}$, where $q$ grows not faster than an expression of the form
$L_Q(\alpha)$. In the following, we assume that there is equality, which
is of course the worst case. The condition can then be rewritten as
$\log q = O((\log Q)^\alpha(\log\log
Q)^{1-\alpha})$ and therefore $k = \log Q / \log q = O((\log Q / \log\log
Q)^{1-\alpha})$. In particular we have $k\leq q+\delta$, so that
Theorem~\ref{thm} can be applied and gives a complexity of $q^{O(\log
k)}$. This yields the following result.

\begin{cor}\label{cor3}
    For finite fields of the form $\F_Q = \F_{q^{2k}}$ where $q$ is
    bounded by $L_Q(\alpha)$, there exists a heuristic algorithm for computing
    discrete logarithms in subexponential time
    $$ L_Q(\alpha)^{O(\log \log Q)}.$$
\end{cor}

This complexity is smaller than $L_Q(\alpha')$ for any $\alpha' >
\alpha$. Hence, for any $\alpha<1/3$, our algorithm is faster than the
best previously known algorithm, namely FFS and its variants.

\section{Main algorithm: proof of Proposition~\ref{prop:onestep}}
\label{sec:descent-one-step}

The algorithm is essentially the same for proving the two points of
Proposition~\ref{prop:onestep}. The strategy is to find relations between
the given polynomial $P(X)$ and its translates by a constant in
$\F_{q^2}$. Let $D$ be the degree of $P(X)$, that we assume to be at
least 1 and at most $k-1$.

The key to find relations is the {\em systematic equation}:
\begin{equation}\label{eq:frobenius}
    X^q-X=\prod_{a\in \F_q}(X-a)\text.
\end{equation}

We like to view Equation~\eqref{eq:frobenius} as involving
the projective line $\bP^1(\F_q)$. Let $\cS=\{(\alpha,\beta)\}$ be a set
of representatives of the $q+1$ points $(\alpha:\beta)\in\bP^1(\F_q)$,
chosen adequately so that the following equality holds.
\begin{equation}
    \label{eq:frobenius-proj}
    X^qY-XY^q=\prod_{(\alpha,\beta)\in\cS}(\beta X-\alpha Y)\text.
\end{equation}

To make translates of $P(X)$ appear, we consider the action of {\em
homographies}. 
Any matrix $m = \begin{pmatrix}a & b\\ c& d\end{pmatrix}$ acts on $P(X)$
with the following formula:
$$m\cdot P = \frac{aP+b}{cP+d}.$$
In the following, this action will become trivial if the matrix $m$ has
entries that are defined over $\F_q$. This is also the case if $m$
is non-invertible. Finally, it is clear that multiplying all the
entries of $m$ by a non-zero constant does not change its action on
$P(X)$. Therefore the matrices of the homographies that we consider are
going to be taken in the following set of cosets:
$$ \cP_q = \PGL(\F_{q^2}) / \PGL(\F_q).$$
(Note that in general $\PGL_2(\F_q)$ is not a
normal subgroup of $\PGL_2(\F_{q^2})$, so that $\cP_q$ is not a quotient
group.) 

To each element $m = \begin{pmatrix}a & b\\ c& d\end{pmatrix}\in
\cP_q$, we associate the equation~\eqref{eq:Em} obtained by substituting $aP+b$
and $cP+d$ in place of $X$ and $Y$ in 
Equation~\eqref{eq:frobenius-proj}.
\def\xfunc{\mathop{\raise-.0125ex\hbox{x}}}
\begin{align*}
    \tag{$E_m$}\label{eq:Em}
(aP+b)^q(cP+d) - (aP+b)(cP+d)^q & =
    \prod_{(\alpha,\beta)\in\cS} \beta(aP+b) - \alpha(cP+d) \\
 & =\prod_{(\alpha,\beta)\in\cS}
       (-c\alpha + a\beta)  P - (d\alpha - b\beta) \\
 & =\lambda\prod_{(\alpha,\beta)\in\cS}
       P - \xfunc(m^{-1} \cdot (\alpha:\beta))\text.
\end{align*}
This sequence of formulae calls for a short comment because of an abuse
of notation in the last expression. First, $\lambda$ is the constant in
$\F_{q^2}$ which makes the leading terms of the two sides match. Then,
the term $P-\xfunc(m^{-1} \cdot
(\alpha:\beta))$ denotes $P-u$ when $m^{-1} \cdot (\alpha:\beta)=(u:1)$
(whence we have $u=\frac{d\alpha - b\beta}{-c\alpha + a\beta}$), or $1$
if $m^{-1} \cdot (\alpha:\beta)=\infty$. The latter may occur since when
$a/c$ is in $\F_q$, the expression $-c\alpha + a\beta$ vanishes for a
point $(\alpha:\beta)\in\bP^1(\F_{q})$ so that one of the factors of the
product contains no term in $P(X)$.
 
Hence the right-hand side of Equation~\eqref{eq:Em} is, up to a
multiplicative constant, a product of $q+1$ or $q$ translates of the
target $P(X)$ by elements of
$\F_{q^2}$. The equation obtained is actually related to the set of
points $m^{-1}\cdot\bP^1(\bF_q)\subset \bP^1(\bF_{q^2})$.
\medskip

The polynomial on the left-hand side of~\eqref{eq:Em} can be rewritten as
a smaller degree equivalent. For this, we use the
special form of the defining polynomial: in $K$ we have $X^q \equiv
\frac{h_0(X)}{h_1(X)}$.  Let us denote by $\tilde{a}$ the element $a^q$ when $a$
is any element of $\F_{q^2}$. Furthermore, we write
$\tilde{P}(X)$ the polynomial $P(X)$ with all its coefficients
raised to the power $q$. The left-hand side of~\eqref{eq:Em} is
$$(\tilde{a}\tilde{P}(X^q)+\tilde{b})(cP(X)+d)
- (aP(X) + b)(\tilde{c}\tilde{P}(X^q)+\tilde{d}),$$
and using the defining equation for the field $K$, it is congruent to
$$
\cL_m \coloneq \left(\tilde{a}\tilde{P}\left(\frac{h_0(X)}{h_1(X)}\right)+\tilde{b}\right)(cP(X)+d)
- (aP(X) +
b)\left(\tilde{c}\tilde{P}\left(\frac{h_0(X)}{h_1(X)}\right)+\tilde{d}\right).
$$
The denominator of $\cL_m$ is a
power of~$h_1$ and its numerator has degree at most $(1+\delta) D$ where
$\delta=\max(\deg h_0,\deg h_1)$. We say that $m\in\cP_q$ yields a
relation if this numerator of $\cL_m$ is $\lceil D/2 \rceil$-smooth. 

To any $m\in\cP_q$, we associate a row vector $v(m)$ of dimension $q^2+1$ in the following
way. Coordinates are indexed by $\mu\in\bP^1(\F_{q^2})$, and the value
associated to $\mu\in\F_{q^2}$ is $1$ or $0$ depending on whether
$P-\xfunc(\mu)$ appears in the right-hand side of Equation~\eqref{eq:Em}. Note that
exactly $q+1$ coordinates are $1$ for each $m$. Equivalently, we may write
\begin{equation}\label{eq:v(m)}
v(m)_{\mu\in\bP^1(\F_{q^2})}=\left\{
    \begin{array}{l}
        1\text{ if }\mu=m^{-1}\cdot(\alpha:\beta) \text{ with }
	    (\alpha:\beta)\in\bP^1(\F_q),\\
        0\text{ otherwise}.
    \end{array}
\right.
\end{equation}

We associate to the polynomial $P$ a matrix $H(P)$ whose rows are
the vectors $v(m)$ for which $m$ yields a relation, taking at most one
matrix $m$ in each coset of $\cP_q$. The validity of Proposition~\ref{prop:onestep} crucially relies on the
following heuristic.

\begin{heuristic}\label{heu:fullrank}
    For any $P(X)$, the set of rows $v(m)$ for cosets
$m\in\cP_q$ that yield a relation form a matrix which has full rank $q^2+1$.
\end{heuristic}

As we will note in Section~\ref{sec:heur}, the matrix $H(P)$ is
heuristically expected to have $\Theta(q^3)$ rows, where the implicit
constant depends on $\delta$. This means that for our decomposition
procedure to work, we rely on the fact that $q$ is large enough
(otherwise $H(P)$ may have less than $q^2+1$ rows, which precludes the
possibility that it have rank $q^2+1$).
\medskip

The first point of Proposition~\ref{prop:onestep}, where we descend a
polynomial $P(X)$ of degree $D$ at least 2, follows by linear algebra
on this matrix.
Since we assume that the matrix has full rank, 
then the vector $(\ldots,0,1,0,\ldots)$ with $1$
corresponding to $P(X)$ can be written as a linear combination of the rows.
When doing this linear combination on the equations~\eqref{eq:Em} corresponding to
$P$ we write $\log P(X)$ as a linear combination of $\log P_i$ where
$P_i(x)$ are the elements occurring in the
left-hand sides of the equations. Since there are $O(q^2)$ columns, the elimination process
involves at most $O(q^2)$ rows, and since each row corresponds to
an equation~\eqref{eq:Em}, it involves at most $\deg \cL_m\leq (1+\delta)D$ polynomials in the
left-hand-side\footnote{This estimate of the number of irreducible
  factors is a pessimistic upper bound. In practice, one expects to
  have only $O(\log D)$ factors on average. Since the crude estimate
  does not change the overall complexity, we keep it that way to avoid
  adding another heuristic.}.  In total, the polynomial $D$ is
expressed by a linear combination of at most $O(q^2D)$ polynomials of
degree less than $\lceil D/2\rceil$. The logarithm of $h_1(X)$ is also
involved, as a denominator of $\cL_m$. We have not made precise the
constant in $\F_{q^2}^*$ which occurs to take care of the leading
coefficients. Since discrete logarithms in $\F_{q^2}^*$ can certainly be
computed in polynomial time in $q$, this is not a problem.

Since the order of $\PGL_2(\F_{q^i})$ is $q^{3i}-q^i$, the set of cosets
$\cP_q$ has $q^3+q$ elements. For each $m \in\cP_q$, testing
whether~\eqref{eq:Em}
yields a relation amounts to some polynomial manipulations and a 
smoothness test. All of them can be done in polynomial time in $q$ and
the degree of $P(X)$ which is bounded by $k$. Finally, the linear algebra
step can be done in $O(q^{2\omega})$ using asymptotically fast matrix
multiplication algorithms, or alternatively $O(q^5)$ operations using
sparse matrix techniques.
Indeed, we have $q+1$
non-zero entries per row and a size of $q^2+1$.
Therefore, the overall cost is polynomial in $q$ and $k$ as claimed.
\medskip

For the second part of Proposition~\ref{prop:onestep} we replace $P$ by $X$ during the
construction of the matrix. In that case, both sides of the
equations~\eqref{eq:Em} involve only linear polynomials.
Hence we obtain a linear system whose unknowns
are $\log (X+a)$ with $a\in\F_{q^2}$. Since Heuristic~\ref{heu:fullrank}
would give us only the full rank of the system corresponding to the
right-hand sides of the equations~\eqref{eq:Em}, we have to rely on a
specific heuristic for this step:
\begin{heuristic}\label{heu:linfullrank}
    The linear system constructed from all the equations~\eqref{eq:Em}
    for $P(X)=X$ has full rank.
\end{heuristic}
Assuming that this heuristic holds, we can
solve the linear system and obtain the discrete logarithms of the linear
polynomials and of $h_1(X)$. 

\section{Supporting the heuristic argument in the proof}
\label{sec:heur}

For Heuristic~\ref{heu:fullrank}, we propose two approaches to support this
heuristic. Both allow to gain some confidence in the validity of
the heuristic, but of course none affect the heuristic nature of this
statement.

For the first line of justification, we denote by $\cH$ the matrix of all
the $\#\cP_q=q^3+q$ vectors $v(m)$ defined as in
Equation~\eqref{eq:v(m)}. Associated to a polynomial~$P$,
Section~\ref{sec:descent-one-step} defines the matrix
$H(P)$ formed of the
rows $v(m)$ such that the numerator of $\cL_m$ is smooth. We will give
heuristics that $H(P)$ has $\Theta(q^3)$ rows and then prove that $\cH$ has
rank $q^2+1$, which of course does not prove that its submatrix $H(P)$ has full rank. 

In order to estimate the number of rows of $H(P)$ we assume that the
numerator of $\cL_m$ has
the same probability to be $\lceil \frac{D}{2}\rceil$-smooth as a random
polynomial of same degree. In this paragraph, we assume that the
degrees of $h_0$ and $h_1$ are bounded by $2$, merely to avoid awkward
notations; the result holds for any constant bound $\delta$.
The degree of the numerator of $\cL_m$ is then bounded by $3D$, so we have
to estimate
the probability that a polynomial in $\F_{q^2}[X]$ of degree $3D$ is
$\lceil \frac{D}{2}\rceil$-smooth. For any prime power $q$ and integers
$1\leq m\leq n$, we denote by $N_q(m,n)$ the number of $m$-smooth monic
polynomials of degree $n$. Using analytic methods, Panario et
al. gave a precise estimate of this quantity (Theorem~$1$
of~\cite{FGP98}):
\begin{equation}\label{eq:Flajolet}
	N_{q}(n,m)=q^n \rho\left(\frac{n}{m}\right)\left(1+O\left(\frac{\log
n}{m}\right)\right),
\end{equation}
where $\rho$ is Dickman's function defined as the unique continuous
function such that $\rho(u)=1$ on $[0,1]$ and $u\rho'(u)=\rho(u-1)$ for
$u>1$. We stress that the constant $\kappa$ hidden in the $O()$ notation
is independent of $q$.
In our case, we are interested in the value of $N_{q^2}(3D, \lceil
\frac{D}{2}\rceil)$. Let us call $D_0$ the least
integer such that $1+\kappa\left(\frac{\log (3D)}{\lceil D/2\rceil}\right)$
is at least $1/2$. For $D>D_0$, we will use the
formula~\eqref{eq:Flajolet}; and for $D\le D_0$, we will use the crude
estimate $N_q(n,m) \ge N_q(n,1) = q^n/n!$. Hence the smoothness
probability of $\cL_m$ is at least
$\min\left(\frac{1}{2}\rho(6),1/(3D_0)!\right)$.

More generally, if $\deg h_0$ and $\deg h_1$ are bounded by a constant
$\delta$ then we have a smoothness probability of $\rho(2\delta+2)$ times
an absolute constant. Since we have $q^3+q$ candidates and a constant
probability of success, $H(P)$ has $\Theta(q^3)$ rows.


Now, unless some theoretical obstruction
occurs, we expect a matrix over $\F_\ell$ to have full rank with
probability at least $1-\frac{1}{\ell}$. The matrix $\cH$ is however peculiar, and does enjoy
regularity properties which are worth noticing.
For instance, we have the following proposition.
\begin{prop}
    \label{prop:bigmat-fullrank}
Let $\ell$ be a prime not dividing $q^3-q$. Then the matrix
$\cH$ over $\F_\ell$ has full rank $q^2+1$.
\end{prop}
\begin{proof}
    We may obtain this result in two ways. First, 
\cH\ is
the incidence matrix of a $3-(q^2+1,q+1,1)$ combinatorial design called
\emph{inversive plane} (see e.g.~\cite[Theorem 9.27]{Stinson03}). As such
we obtain the identity $$\cH^T\cH=(q+1)(J_{q^2+1}-(1-q)I_{q^2+1})$$
(see~\cite[Theorem 1.13 and Corollary 9.6]{Stinson03}), where
$J_n$ is the $n\times n$ matrix with all entries equal to one, and $I_n$
is the $n\times n$ identity matrix. This readily gives the result exactly
as announced.

    We also provide an elementary proof of the Proposition.
    We have a
    bijection between rows of $\cH$ and the different possible image
    sets of the projective line $\bP^1(\F_q)$ within $\bP^1(\F_{q^2})$,
    under injections of the form $(\alpha:\beta)\mapsto
    m^{-1}\cdot(\alpha:\beta)$. All these $q^3+q$ image sets have size
    $q+1$, and by symmetry all points of $\bP^1(\F_{q^2})$ are reached
    equally often. Therefore, the sum of all rows of $\cH$ is the
    vector whose coordinates are all equal to $\frac1{1+q^2}(q^3+q)(q+1)=q^2+q$.

    Let us now consider the sum of the rows in $\cH$ whose first
    coordinate is $1$ (as we have just shown, we have $q^2+q$ such rows).
    Those correspond to image sets of $\bP^1(\F_q)$ which contain one
    particular point, say $(0:1)$. The value of the sum for any other
    coordinate indexed by e.g.\ $Q\in\bP^1(\F_{q^2})$ is the number of
    image sets $m^{-1}\cdot\bP^1(\F_q)$ which contain both $(0:1)$ and
    $Q$, which we prove is equal to $q+1$ as follows. Without loss of generality, we may assume $Q=\infty=(1:0)$.
    We need to count the relevant
    homographies $m^{-1}\in\PGL_2(\F_{q^2})$, modulo
    $\PGL_2(\F_q)$-equivalence $m\equiv hm$. By
    $\PGL_2(\F_q)$-equivalence, we may without loss of generality assume
    that $m^{-1}$ fixes $(0:1)$ and $(1:0)$.
    Letting $m^{-1}=
\begin{pmatrix}a&b\\c&d\end{pmatrix}$, we obtain $(b:d)=(0:1)$ and
    $(a:c)=(1:0)$, whence $b=c=0$, and both $a,d\not=0$. We may
    normalize to $d=1$, and notice that multiplication of $a$ by a scalar
    in $\F_q^*$ is absorbed in $\PGL_2(\F_q)$-equivalence. Therefore
    the number of suitable $m$ is $\#{\F_{q^2}^*}/{\F_q^*}=q+1$.

    These two facts show that the row span of $\cH$
    contains the vectors $(q^2+q, \ldots, q^2+q)$ and $(q^2+q, q+1,
    \ldots, q+1)$. The vector $(q^3-q,0,\ldots,0)$ is obtained as a linear
    combination of these two vectors, which suffices to prove that
    $\cH$ has full rank, since the same reasoning holds
    for any coordinate.

\end{proof}

Proposition~\ref{prop:bigmat-fullrank}, while encouraging, is clearly not
sufficient. We are, at the moment, unable to provide a proof of a
more useful statement. On the experimental side, it is reasonably easy to
sample arbitrary subsets of the rows of $\cH$ and check for their rank.
To this end, we propose the following experiment. We have considered
small values of $q$ in the range $[16,\ldots,64]$, and made~50 random picks of
subsets $S_i\subset\cP_q$, all of size exactly $q^2+1$. For each we
considered the matrix of the corresponding linear system, which is made of
selected rows of the matrix \cH, and computed its determinant $\delta_i$.
For all values of $q$ considered,
we have observed the following facts.
\begin{itemize}
    \item First, all square matrices considered had full rank over \bZ.
        Furthermore, their determinants had no common factor apart
        possibly from those appearing in the factorization of $q^3-q$ as
        predicted by Proposition~\ref{prop:bigmat-fullrank}. In fact,
        experimentally it seems that only the factors of $q+1$ are
        causing problems.
    \item We also explored the possibility that modulo some primes, the
        determinant could vanish with non-negligible probability. We thus
        computed the pairwise GCD of all~50 determinants computed, for
        each $q$. Again, the only prime factors appearing in the GCDs
        were either originating from the factorization of $q^3-q$, or
        sporadically from the birthday paradox.
\end{itemize}
        \begin{table}
\begin{center}
    \begin{minipage}[t]{0.5\textwidth}
\begin{tabular}{c|c|l|l}
    $q$ & \#trials &  in $\gcd(\{\delta_i\})$ &
        in $\gcd(\delta_i, \delta_j)$\\
    \hline
16 & 50 & 17 & 691\\
17 & 50 & 2, 3 & 431, 691\\
19 & 50 & 2, 5 & none above $q^2$\\
23 & 50 & 2, 3 & none above $q^2$\\
25 & 50 & 2, 13 & none above $q^2$\\
27 & 50 & 2, 7 & 1327\\
29 & 50 & 2, 3, 5 & none above $q^2$\\
31 & 50 & 2 & 1303, 3209\\
32 & 50 & 3, 11 & none above $q^2$\\

\end{tabular}
    \end{minipage}%
    \begin{minipage}[t]{0.5\textwidth}
\begin{tabular}{c|c|l|l}
    $q$ & \#trials &  in $\gcd(\{\delta_i\})$ &
        in $\gcd(\delta_i, \delta_j)$\\
    \hline
37 & 50 & 2, 19 & 2879\\
41 & 50 & 2, 3, 7 & none above $q^2$\\
43 & 50 & 2, 11 & none above $q^2$\\
47 & 50 & 2, 3 & none above $q^2$\\
49 & 50 & 2, 5 & none above $q^2$\\
53 & 50 & 2, 3 & none above $q^2$\\
59 & 50 & 2, 3, 5 & none above $q^2$\\
61 & 50 & 2, 31 & none above $q^2$\\
64 & 50 & 5, 13 & none above $q^2$\\
\end{tabular}
    \end{minipage}%

\caption{\label{tab:experiment1}Prime factors appearing in determinant of
random square submatrices of \cH\ (for one given set of random trials)}
\end{center}
        \end{table}
        These results are 
summarized in
table~\ref{tab:experiment1}, where the last column omits small prime
factors below $q^2$.
Of course, we remark that considering square submatrices is a more demanding check than
what Heuristic~\ref{heu:fullrank} suggests, since our algorithm only
needs a slightly larger matrix of size $\Theta(q^3)\times(q^2+1)$ to have
full rank.
\medskip

A second line of justification is more direct and natural, as it is
possible to implement the algorithm outlined in
Section~\ref{sec:descent-one-step}, and verify that it does provide the
desired result. A \textsc{Magma} implementation validates this claim, and
has been used to implement descent steps for an example field of
degree~$53$ over $\F_{53^2}$. An example step in this context is given
for applying our algorithm to a polynomial of degree~10, attempting to
reduce it to polynomials of degree~6 or less.
Among the 148,930 elements
of $\cP_q$, it sufficed to consider only 71,944 matrices $m$, of which about 3.9\%
led to relations, for a minimum sufficient number of relations equal to
$q^2+1=2810$ (as more than half of the elements of $\cP_q$ had not even
been examined at this point, it is clear that getting more relations was
easy---we did not have to).
As the defining polynomial
for the finite field considered was constructed with $\delta=\deg
h_{0,1}=1$, all left-hand sides involved
had degree 20.
The polynomials appearing in their
factorizations had the
following degrees (the number in brackets give the number of distinct
polynomials found for each degree): 1(2098), 2(2652), 3(2552), 4(2463), 5(2546), 6(2683).
Of course this tiny example size uses no
optimization, and is only intended to check the validity of
Proposition~\ref{prop:onestep}.

\bigskip

As for Heuristic~\ref{heu:linfullrank}, it is already present
in~\cite{Joux13} and~\cite{GGMZ13}, so this is not a new heuristic.
Just like for Heuristic~\ref{heu:fullrank}, it is based on the fact that
the probability that a left-hand side is $1$-smooth and yields a relation
is constant. Therefore, we have a system with $\Theta(q^3)$ relations
between $O(q^2)$ indeterminates, and it seems reasonable to expect that
it has full rank. On the other hand, there is not as much algebraic
structure in the linear system as in Heuristic~\ref{heu:fullrank}, so that
we see no way to support this heuristic apart from testing it on several
inputs. This was already done (including for record computations)
in~\cite{Joux13} and~\cite{GGMZ13}, so we do not elaborate on our own
experiments that confirm again that Heuristic~\ref{heu:linfullrank} seems
to be valid except for tiny values of $q$.

\paragraph{An obstruction to the heuristics.}

As noted by Cheng, Wan and Zhuang~\cite{traps13}, the irreducible factors
of $h_1X^q-h_0$ other than the degree $k$ factor that is used to define
$\bF_{q^{2k}}$ are problematic. Let $P$ be such a problematic polynomial.
The fact that it divides the defining equation implies that it also
divides the $\cL_m$ quantity that is involved when trying to build a
relation that relates $P$ to other polynomials. Therefore the first part
of Proposition~\ref{prop:onestep} can not hold for this $P$. Similarly, if
$P$ is linear, its presence will prevent the second part of
Proposition~\ref{prop:onestep} to hold since the logarithm of $P$ can not
be found with the technique of Section~\ref{sec:descent-one-step}.
We present here a technique to deal with the problematic polynomials.
(The authors of~\cite{traps13} proposed another solution to keep the
quasi-polynomial nature of algorithm.)

\begin{prop}\label{solvetrap}
For each problematic polynomial $P$ of degree $D$, we can find a linear relation
between $\log P$, $\log h_1$ and $O(D)$ logarithms of polynomials of degree at
most $(\delta-1)D$ which are not problematic.
\end{prop}

\begin{proof}
Let $P$ be an irreducible factor of $h_1X^q-h_0$ of degree $D$.  Let us
consider $P^q$; by reducing modulo $h_1X^q-h_0$ and clearing
denominators, there exists a polynomial $A(X)$ such that
\begin{equation}\label{eq:freerel}
h_1^D P^q = h_1^D\tilde{P}\left(\frac{h_0}{h_1}\right)
+(h_1X^q-h_0)A(X).
\end{equation}
Since $P$ divides two of the terms of this equality, it must also
divide the third one, namely the polynomial $\cR =
h_1^D\tilde{P}\left(h_0/h_1\right)$. Let $v_P\ge 1$ be the valuation of
$P$ in $\cR$. In the finite field $\bF_{q^{2k}}$ we obtain the following
equalities between logarithms:
\begin{equation*}
	(q-v_P)\log P = -D\log h_1 +\sum_i e_i \log Q_i,
\end{equation*}
where $Q_i$ are the irreducible factors of $\cR$ other than $P$ and $e_i$
their valuation in~$\cR$. A polynomial $Q_i$ can not be problematic.
Otherwise, it would divide the right-hand side of
Equation~\eqref{eq:freerel}, and therefore, also the left-hand side, which
is impossible.
Since $v_P\leq \frac{\deg
\cR}{\deg P}\leq \delta <q$, the quantity $q-v_P$ is invertible modulo
$\ell$ (we assume, as usual that $\ell$ is larger than $q$) and we obtain
a relation between $\log P$, $\log h_1$ and the
logarithms of the non-problematic polynomials $Q_i$.
The degree of $\cR/P^{v_P}$ is at most $(\delta-1)D$, which gives the
claimed bound on the degrees of the $Q_i$.
\end{proof}

If $\delta\le 2$, this proposition solves the issues raised
by~\cite{traps13} about
problematic polynomials. Indeed, for each problematic polynomial of
degree $D>1$, it will be possible to rewrite its logarithm in terms of
logarithms of non-problematic polynomials of at most the same degree that
can be descended in the usual way. Similarly, each problematic
polynomial of degree 1 can have its logarithm rewritten in terms of the
logarithms of other non-problematic linear polynomials. Adding these
relations to the ones obtained in Section~\ref{sec:descent-one-step}, we
expect to have a full-rank linear system.

If $\delta>2$, we need to rely on the additional heuristic. Indeed, when
descending the $Q_i$ that have a degree potentially larger than the
degree of $D$, we could hit again the problematic polynomial we started
with, and it could be that the coefficients in front of $\log P$ in the
system vanishes. More generally, taking into account all the problematic
polynomials, if when we apply Proposition~\ref{solvetrap} to them we get
polynomials $Q_i$ of higher degrees, it could be that descending those we
creates loops so that the logarithms of some of the problematic
polynomials could not be computed. We expect this event to be very
unlikely. Since in all our experiments it was always possible to obtain
$\delta=2$, we did not investigate further.

\paragraph{Finding appropriate $h_0$ and $h_1$.}

One key fact about the algorithm is the existence of two polynomials
$h_0$ and $h_1$ in $\F_{q^2}[X]$ such that $h_1(X)X^q-h_0(X)$ has an
irreducible factor of degree $k$. A partial solution is due to
Joux~\cite{Joux13} who showed how to construct such polynomials when
$k\in\{q-1,q,q+1\}$. 
No such deterministic construction is known in the general
case, but experiments show that one can apparently choose $h_0$ and $h_1$ of
degree at most $2$. We performed an experiment for every odd prime
power $q$ in $[3,\ldots,1000]$ and every $k\leq q$ and found that we could
select $a\in\F_{q^2}$ such that $X^q+X^2+a$ has an irreducible factor
of degree $k$. Finally, note that the result is similar to a commonly made
heuristic in discrete logarithm algorithms: for fixed
$f\in\F_{q^2}[X,Y]$ and random $g\in\F_{q^2}[X,Y]$, the polynomial
$\text{Res}_Y(f,g)$ behaves as a random polynomial of same degree with
respect to the degrees of its irreducible factors.

\section{Some directions of improvement}
\label{sec:improvement}
The algorithm can be modified in several ways. On the one hand one can
obtain a better complexity if one proves
a stronger result on the smoothness probability. On the other
hand, without changing the complexity, one can obtain a version which
should behave better in practice. 

\subsection{Complexity improvement}
Heuristic~\ref{heu:fullrank} tells that a rectangular matrix with $\Theta(q)$
times more rows than columns has full rank. It seems reasonable to expect
that only a constant times more rows than columns would be enough to get
the full rank properties (as is suggested by the experiments proposed in
Section~\ref{sec:heur}). Then, it means that we expect to have a lot of
choices to select the best relations, in the sense that their left-hand
sides split into irreducible factors of degrees as small as possible.

On average, we expect to be able to try $\Theta(q)$ relations for each row
of the matrix. So, assuming that the numerators of $\cL_m$ behave like
random polynomials of similar degrees, we have to evaluate the
expected smoothness that we can hope for after trying $\Theta(q)$
polynomials of degree $(1+\delta)D$ over $\F_{q^2}$. Set $u=\log q /
\log\log q$, so that $u^u\approx q$. According to~\cite{FGP98} it is then possible to replace $\lceil
D/2\rceil$ in
Proposition~\ref{prop:onestep} by the value $O(D\log\log q/\log q)$.

Then, the discussion leading to Theorem~\ref{thm} can be changed to
take this faster descent into account. We keep the same estimate for the
arity of each node in the tree, but the depth is now only in $\log k /
\log\log q$. Since this depth ends up in the exponent, the resulting
complexity in Theorem~\ref{thm} is then
$$ \max(q, k)^{O(\log k / \log\log q)}.$$

\subsection{Practical improvements}
Because of the arity of the descent tree, the breadth eventually
exceeds the number of polynomials below some degree bound. It
makes no sense, therefore, to use the descent procedure beyond
this point, as the recovery of discrete logarithms of all these
polynomials is better achieved as a pre-computation. Note that this
corresponds to the computations of the $L(1/4+\epsilon)$ algorithm which starts by
pre-computing the logarithms of polynomials up to degree $2$.
In our case, we could in principle go up to degree $O(\log q)$ without
changing the complexity.
\medskip

We propose another practical improvement in the case where we would like
to spend more time descending a given polynomial $P$ in order to improve
the quality of the descent tree rooted at $P$. 
The set of polynomials appearing in the right-hand side of
Equation~\eqref{eq:Em} in Section~\ref{sec:descent-one-step} is
$\{P-\lambda\}$, because in the factorization of $X^q-X$, we
substitute $X$ with $m\cdot P$ for homographies~$m$. In fact, we
may apply $m$ to $(P:P_1)$ for any polynomial $P_1$ whose degree
does not exceed that of $P$. In the right-hand sides, we will have only
factors of form $P - \lambda P_1$ for $\lambda$ in $\F_{q^2}$.
On the left-hand sides, we have polynomials of the same degree as before,
so that the smoothness probability is expected to be the same.
Nevertheless, it is possible to test several $P_1$ polynomials, and to
select the one that leads to the best tree.

This strategy can also be useful in the following context (which will not
occur for large enough $q$):
it can
happen that for some triples $(q,D,D')$ one has $N_{q^2}(3D,D')/q^n\approx 1/q$. In
this case we have no certainty that we can descend a degree-$D$
polynomial to degree $D'$, but we can hope that at least one of the
$P_1$ allows to descend.

Finally, if one decides to use several auxiliary $P_1$ polynomials to descend
a polynomial $P$, it might be interesting to take a set of polynomials
$P_1$ with an arithmetic structure, so that the smoothness tests on the
left-hand sides can benefit from a sieving technique.

\section{Conclusion}
The algorithm presented in this article achieves a significant improvement
of the asymptotic complexity of discrete logarithm in finite fields, in
almost the whole range of parameters where the Function Field Sieve was
presently the most competitive algorithm. Compared to existing
approaches, and in particular to the line of recent
works~\cite{Jo13faster,GGMZ13}, the practical relevance of our algorithm
is not clear, and will be explored by further work. 

We note that the analysis of the algorithm presented here is heuristic, as discussed
in Section~\ref{sec:heur}. Some of the heuristics we stated,
related to the properties of matrices $H(P)$ extracted from the
matrix $\cH$, seem accessible to more solid justification. It seems
plausible to have the validity of algorithm rely on the sole heuristic of
the validity of the smoothness estimates.

The crossing point between the $L(1/4)$ algorithm and our quasi-polynomial
one is not determined yet. One of the key factors which hinders the practical
efficiency of this algorithm is the $O(q^2D)$ arity of the descent tree,
compared to the $O(q)$ arity achieved by techniques based on Gröbner
bases~\cite{Jo13faster} at the expense of a $L(1/4+\epsilon)$ complexity.
Adj et al.~\cite{AMOR13} proposed to mix the two algorithms
and deduced that the new descent technique must be used for cryptographic
sizes. Indeed, by estimating the time required to
compute discrete logarithms in $\F_{3^{6\cdot 509}}$, they showed the weakness
of some pairing-based cryptosystems. 

\ifanon
\else
\section*{Acknowledgements}
The authors would like to thank Daniel J. Bernstein for his comments on
an earlier version of this work, and for pointing out to us the possible
use of asymptotically fast linear algebra for solving the linear systems
encountered.

\fi

\bibliographystyle{splncs03}
\bibliography{article}

\begin{thebibliography}{10}
\providecommand{\url}[1]{\texttt{#1}}
\providecommand{\urlprefix}{URL }

\bibitem{AMOR13}
Adj, G., Menezes, A., Oliveira, T., Rodr{\'i}guez-Henr{\'i}quez, F.: Weakness
  of $\mathbb{F}_{3^{6\cdot509}}$ for discrete logarithm cryptography.
  Cryptology ePrint Archive, Report 2013/446 (2013),
  \url{http://eprint.iacr.org/2013/446/}

\bibitem{Adl79}
Adleman, L.: A subexponential algorithm for the discrete logarithm problem with
  applications to cryptography. In: Foundations of Computer Science, 1979.,
  20th Annual Symposium on. pp. 55--60. IEEE (1979)

\bibitem{Adl94}
Adleman, L.: The function field sieve. In: Algorithmic number theory -- ANTS I.
  Lecture Notes in Comput. Sci., vol. 877, pp. 108--121. Springer (1994)

\bibitem{BlFuMuVa84}
Blake, I.F., Fuji-Hara, R., Mullin, R.C., Vanstone, S.A.: Computing logarithms
  in finite fields of characteristic two. SIAM J. Alg. Disc. Meth.  5(2),
  276--285 (Jun 1984)

\bibitem{traps13}
Cheng, Q., Wan, D., Zhuang, J.: Traps to the {BGJT}-algorithm for discrete
  logarithms. Cryptology ePrint Archive, Report 2013/673 (2013),
  \url{http://eprint.iacr.org/2013/673/}

\bibitem{Cop84}
Coppersmith, D.: Fast evaluation of logarithms in fields of characteristic two.
  IEEE Transactions on Information Theory  30(4),  587--594 (1984)

\bibitem{DiHe76}
Diffie, W., Hellman, M.: New directions in cryptography. IEEE Transactions on
  Information Theory  22(6),  644--654 (1976)

\bibitem{record1971}
G{\"o}loglu, F., Granger, R., McGuire, G., Zumbr{\"a}gel, J.: Discrete
  logarithm in {{GF}}$(2^{1971})$ (Feb 2013), announcement to the NMBRTHRY list

\bibitem{record6120}
G{\"o}loglu, F., Granger, R., McGuire, G., Zumbr{\"a}gel, J.: Discrete
  logarithm in {{GF}}$(2^{6120})$ (Apr 2013), announcement to the NMBRTHRY list

\bibitem{GGMZ13}
G{\"o}loglu, F., Granger, R., McGuire, G., Zumbr{\"a}gel, J.: On the {F}unction
  {F}ield {S}ieve and the impact of higher splitting probabilities. In:
  Advances in Cryptology -- CRYPTO 2013. Lecture Notes in Comput. Sci., vol.
  8043, pp. 109--128. Springer (2013)

\bibitem{Gor93}
Gordon, D.M.: Discrete logarithms in {{GF(p)}} using the number field sieve.
  SIAM Journal on Discrete Mathematics  6(1),  124--138 (1993)

\bibitem{record1778}
Joux, A.: Discrete logarithm in {{GF}}$(2^{1778})$ (Feb 2013), announcement to
  the NMBRTHRY list

\bibitem{record4080}
Joux, A.: Discrete logarithm in {{GF}}$(2^{4080})$ (Mar 2013), announcement to
  the NMBRTHRY list

\bibitem{record6168}
Joux, A.: Discrete logarithm in {{GF}}$(2^{6168})$ (May 2013), announcement to
  the NMBRTHRY list

\bibitem{Jo13faster}
Joux, A.: Faster index calculus for the medium prime case. {A}pplication to
  1175-bit and 1425-bit finite fields. In: {Advances in Cryptology -- EUROCRYPT
  2013}, Lecture Notes in Comput. Sci., vol. 7881, pp. 177--193. Springer
  (2013)

\bibitem{Joux13}
Joux, A.: A new index calculus algorithm with complexity {{L}}$(1/4+o(1))$ in
  very small characteristic. Cryptology ePrint Archive, Report 2013/095 (2013)

\bibitem{JoLe06}
Joux, A., Lercier, R.: {The function field sieve in the medium prime case}. In:
  Advances in Cryptology -- EUROCRYPT 2006. Lecture Notes in Comput. Sci., vol.
  4005, pp. 254--270. Springer (2006)

\bibitem{JLVS07}
Joux, A., Lercier, R., Smart, N., Vercauteren, F.: The number field sieve in
  the medium prime case. In: Advances in Cryptology -- CRYPTO 2006, Lecture
  Notes in Comput. Sci., vol. 4117, pp. 326--344. Springer (2006)

\bibitem{FGP98}
Panario, D., Gourdon, X., Flajolet, P.: An analytic approach to smooth
  polynomials over finite fields. In: Algorithmic number theory -- ANTS~III,
  Lecture Notes in Comput. Sci., vol. 1423, pp. 226--236. Springer (1998)

\bibitem{PoHe78}
Pohlig, S., Hellman, M.: An improved algorithm for computing logarithms over
  ${GF}(p)$ and its cryptographic signifiance. IEEE Transactions on Information
  Theory  24(1),  106--110 (1978)

\bibitem{Stinson03}
Stinson, D.R.: Combinatorial designs : constructions and analysis. Springer
  (2003)

\end{thebibliography}

\end{document}